\documentclass[11pt,a4paper,reqno]{amsart}
\usepackage{amsmath}
\usepackage[english]{babel}
\usepackage{latexsym}
\usepackage{amssymb}
\usepackage{amscd}
\usepackage{amsgen,amstext,amsbsy,amsopn}
\usepackage{math rsfs}
\usepackage{bm,bbm}
\usepackage{amsthm,epsfig,graphicx,graphics}
\usepackage[latin1]{inputenc}
\usepackage{xspace}
\usepackage{amsxtra}
\usepackage{xcolor}
\usepackage{dsfont}
\usepackage{enumerate}

\usepackage[pdftex]{hyperref}

\usepackage{geometry}
\numberwithin{equation}{section}
\newcommand{\bdm}{\begin{displaymath}}
\newcommand{\edm}{\end{displaymath}}
\newcommand{\bdn}{\begin{eqnarray}}
\newcommand{\edn}{\end{eqnarray}}
\newcommand{\bay}{\begin{array}{c}}
\newcommand{\eay}{\end{array}}
\newcommand{\ben}{\begin{enumerate}}
\newcommand{\een}{\end{enumerate}}
\newcommand{\beq}{\begin{equation}}
\newcommand{\eeq}{\end{equation}}
\newcommand{\beqn}{\begin{eqnarray}}
\newcommand{\eeqn}{\end{eqnarray}}
\newcommand{\bml}[1]{\begin{multline} #1 \end{multline}}
\newcommand{\bmln}[1]{\begin{multline*} #1 \end{multline*}}

\newcommand{\mean}[1]{\lf\langle #1 \ri\rangle}
\renewcommand{\leq}{\leqslant}
\renewcommand{\geq}{\geqslant}

\newcommand{\tx}{\textstyle}
\newcommand{\disp}{\displaystyle}

\newcommand{\lf}{\left}
\newcommand{\ri}{\right}

\newcommand{\braket}[2]{\lf\langle #1|#2 \ri\rangle}

\newcommand{\xv}{\mathbf{x}}
\newcommand{\xvp}{\mathbf{x}^{\prime}}
\newcommand{\rv}{\mathbf{r}}
\newcommand{\yv}{\mathbf{y}}
\newcommand{\kv}{\mathbf{k}}

\newcommand{\nv}{\mathbf{n}}
\newcommand{\qv}{\mathbf{q}}
\newcommand{\pv}{\mathbf{p}}
\newcommand{\fv}{\mathbf{f}}

\newcommand{\alv}{\bm{\alpha}}

\newcommand{\diff}{\mathrm{d}}
\newcommand{\eps}{\varepsilon}

\newcommand{\ba}{\mathcal{B}}
\newcommand{\one}{\mathds{1}}

\newcommand{\form}{\F_{\alv(t)}}
\newcommand{\dom}{\mathscr{D}}
\newcommand{\ham}{H_{\alv(t)}}
\newcommand{\phila}{\phi_{\lambda}}
\newcommand{\philas}{\phi_{\lambda,s}}
\newcommand{\la}{\lambda}
\newcommand{\Gammal}{\Gamma_{\la}}

\newcommand{\R}{\mathbb{R}}

\newcommand{\C}{\mathbb{C}}

\newcommand{\F}{\mathcal{F}}

\newcommand{\LL}{\mathcal{L}}

\newcommand{\I}{\mathcal{I}}
\newcommand{\J}{\mathcal{J}}
\newcommand{\KK}{\mathcal{K}}

\newtheorem{teo}{Theorem}[section]
\newtheorem{lem}{Lemma}[section]
\newtheorem{pro}{Proposition}[section]
\newtheorem{cor}{Corollary}[section]

\theoremstyle{remark}
\newtheorem{rem}{Remark}

\begin{document}

\title{Two-dimensional Time-dependent Point Interactions}

\author[R. Carlone]{Raffaele Carlone}
\address{Universit\`{a} ``Federico II'' di Napoli, Dipartimento di Matematica e Applicazioni ``R. Caccioppoli'', MSA, via Cinthia, I-80126, Napoli, Italy.}
\email{raffaele.carlone@unina.it}

\author[M. Correggi]{Michele CORREGGI}
\address{Dipartimento di Matematica e Fisica, Universit\`{a} degli Studi Roma Tre, L.go San Leonardo Murialdo, 1, 00146, Rome, Italy.}
\email{michele.correggi@gmail.com}

\author[R. Figari]{Rodolfo Figari}
\address{Universit\`{a} ``Federico II'' di Napoli, Dipartimento di Fisica e INFN Sezione di Napoli, MSA, I-80126, Napoli, via Cinthia, Italy.}
\email{rodolfo.figari@na.infn.it}

\dedicatory{dedicated to Pavel Exner}

\date{December 20th, 2015}

\begin{abstract} 
	We study the time-evolution of a quantum particle subjected to time-dependent zero-range forces in two dimensions. After establishing  a conceivable ansatz for the solution to the Schr\"{o}dinger equation, we prove that the wave packet time-evolution is completely specified by the solutions of a system of Volterra-type equations -- the {\it charge equations} -- involving the coefficients of the singular part of the wave function, thus extending to the two-dimensional case known results in one and three dimensions.
\end{abstract}

\maketitle

\tableofcontents

\section{Introduction and Main Results}
Point interactions have been an important theoretical tool to investigate non-trivial qualitative features of the evolution of quantum systems. Since the early days of quantum mechanics they have been extensively used to provide solvable models in various fields of applied quantum physics such as solid state physics of perfect and disordered crystals, spectral nuclear structure, low energy neutron-nuclei scattering and many others. 
%
In the last edition of the reference book in the field \cite{albe}, one can find a detailed and updated reading list on the subject.

The main feature of point interaction hamiltonians is that they are characterized by a minimal set of physical parameters. All the information about the spatial  geometry of the interaction potential acting on the quantum particle is included in the set of positions of the scattering centers, whereas the dynamical parameters consist in a set of real numbers characterizing  boundary conditions that any function in the hamiltonian domain has to satisfy at the scattering centers. Moreover,  in one and three dimensions, it was shown that the behavior at the scattering centers at each time is sufficient to determine uniquely the solution of the Schr\"odinger equation at any time and at any point in space. 

Later, it was recognized that a function of time (referred to as {\it charge} in the following) contains all the information about the behavior of the state function around one interaction point and that the charges are solutions of a system of Volterra integral equations. This extreme simplification of the Cauchy problem was used to generalize the theory to time-dependent and nonlinear point interactions. Such models were employed to investigate ionization issues and problems of quantum evolution in presence of concentrated nonlinearities in one and three dimensions \cite{AT,ADFT,CCLR,CD,CDFM,CLR,CFNT1,CFNT,DFT2}.  

The two dimensional problem turned out to be decisively thornier. As an aside, let us mention that in dimension two the laplacian and the formal delta potential scale in the same way under spacial dilation 
\bdm
	 - \Delta_{a \xv} + \lambda \delta \lf( a \xv \ri) = \frac{1}{a^{2}} [ - \Delta_{\xv} + \lambda \delta \lf( \xv \ri)].
\edm
A relevant consequence of such a scale invariance is that, if $E$ is an eigenvalue of the formal hamiltonian, the same must be true for $a E$, for all $a >0$, making the hamiltonian either trivial or non-self-adjoint. In fact, one of the possible way to define a zero-range potential in dimension two is by ``dimensional regularization", a modern quantum field renormalization scheme, which turned out to be very fruitful as a renormalization tool in Yang-Mills theory (see, e.g., \cite{F} for further details).
The extension of the model to real (or complex)  dimensions $2-\epsilon$ supplies a non-scale invariant theory and a dimensional coupling constant. A (coupling constant dependent) limit $\epsilon \rightarrow 0^+$ then provides the same family of hamiltonians obtained using one of the procedures now available to characterize all the self-adjoint extensions of the laplacian restricted to functions supported outside the set of positions of the interaction centers \cite{albe}. 

From the technical point of view the presence of a logarithmic singularity in the fundamental solution of the Laplace equation  $\displaystyle - \Delta K  + \lambda K =  \delta  $  makes the charge equations in dimension two more difficult to deal with. In particular the techniques of fractional integration and differentiation used in dimension three to regularize the equations are no longer available. This is the main reason why the two-dimensional problem was still open, in spite of the progress in the one- and three-dimensional cases.

In the following we investigate the evolution problem generated by a time dependent point interaction hamiltonian in dimension two. We first review notation, definitions and we state our main results. The last section is devoted to the proofs.

\subsection{The model}

In this paper we want to focus on the study of the time-evolution generated by a time-dependent Schr\"{o}dinger operator with two-dimensional point interactions. More precisely the formal expression we start with is the following
\beq
	\label{eq: formal ham}	
	\tilde{H} = - \Delta + \sum_{j = 1}^N \mu_j(t) \delta(\xv - \yv_j),
\eeq
where $ \xv, \yv_j \in \R^2 $, $ \delta(\xv - \yv_j) $ is the Dirac delta distribution supported at point $ \yv_j $ and $ \mu_j (t) \in \R $ is its strength. The expression above is just formal because in two or more dimensions one can not give a rigorous meaning to the Dirac delta potential, not even in the sense of quadratic forms: the difficulty comes from the fact that $ H^1(\R^2) $, the form domain of $ - \Delta $, contains  functions whose value at a given point $ \yv_j $ might not be defined. To circumvent this problem one can follow different procedures, e.g., introduce a symmetric operator \cite{albe} which coincides with \eqref{eq: formal ham} on a suitable subset of $ H^2(\R^2)$  and study its self-adjoint extensions. Alternatively but equivalently it is possible to introduce (see below) a quadratic form associated with \eqref{eq: formal ham} and study its closedness, i.e., for any 
\beq
	\alpha_j(t) \in C^1(\R)
\eeq
with $ \alv(t) = \lf( \alpha_1(t), \ldots, \alpha_N(t) \ri) $,
\bml{
	\label{eq: form}
	\form[\psi] = \int_{\R^2} \diff \rv \lf\{ \lf| \nabla \phila \ri|^2 + \la \lf| \phila \ri|^2 - \la |\psi|^2 \ri\} + \sum_{j = 1}^N \lf( \alpha_j(t) + \tx\frac{1}{2\pi} \log \frac{\sqrt{\la}}{2} - \frac{\gamma}{2\pi} \ri) |q_j|^2	\\
	+ \frac{1}{2\pi} \sum_{j \neq k} q^*_{j} q_{k} K_{0}(\sqrt{\lambda}\lf| \yv_j - \yv_k \ri|).
}
defined on the domain (notice that the domain $ \dom[\F] $ is actually time-independent)
\beq
	\label{eq: form domain}
	\dom[\F] 
	= \bigg\{ \psi \in L^2(\R^2) \: \Big| \: \psi = \phila + \frac{1}{2\pi} \sum_{j = 1}^N q_j K_0(\sqrt{\la} |\xv - \yv_j|), \phila \in H^1(\R^2), q_j \in \C \bigg\}.
\eeq 
Here $K_0(\sqrt{\la} |\xv|)$ denote the anti-Fourier transform of $ (|\kv|^2 + \lambda)^{-1} $  for any $ \lambda > 0 $. $ K_0(\xv) $ is the modified Bessel function of second kind of order $ 0 $ (also known as Macdonald function, see \cite[Sec. 9.6]{AS}). It belongs to  $L^2 (\R^2)$, it is exponentially decreasing for large $|\xv|$ and its asymptotic behavior for small $|\xv|$ reads \cite[Eq. 9.6.13]{AS}
\beq
	\label{eq: k0 asympt}
	K_0(\sqrt{\la} |\xv|) \underset{|\xv| \to 0}{=} - \log \frac{\sqrt{\la} |\xv|}{2} - \gamma + o(1),
\eeq
with $ \gamma $ the Euler's number. Functions in the domain of $ \form $ are thus composed by a {\it regular part} $ \phi $ and a {\it singular part} containing a local singularity proportional to $ - \log|\xv-\yv_j| $, whose coefficient $ q_j $ is the so-called {\it charge} already mentioned above. 

It is easily checked that $ \dom[\F] $ is independent of $ \la $: a simple way to make this apparent is to observe that for any $ \lambda_1, \lambda_2 > 0 $, $ K_0(\sqrt{\la_1} |\xv|) - K_0(\sqrt{\la_2} |\xv|) \in H^1(\R^2) $, as one can easily verify by considering the Fourier transforms. Moreover the quadratic form \eqref{eq: form} defined on \eqref{eq: form domain} is closed and bounded from below as a consequence of the completeness of $ H^1(\R^2) $ and $ \mathbb{C}^N $. In a much more general setting the proof can be found in \cite{DFT}. 

Therefore it defines for any $ t \in \R $ a unique self-adjoint operator $ \ham $, whose domain is 
\bml{
	\label{eq: op domain}
	\dom(\ham) = \bigg\{ \psi \in L^2(\R^2) \: \Big| \: \psi = \phila + \frac{1}{2\pi} \sum_{j=1}^N q_j(t) K_0(\sqrt{\la} |\xv - \yv_j|), \phila \in H^2(\R^2), 	\\
	\disp\lim_{\xv \to \yv_j} \phila(\xv) =  \lf( \Gammal \qv \ri)_j(t) \bigg\},
}
with
\beq
	\label{eq: op gammal}
	\lf( \Gammal \ri)_{jk} =
	\begin{cases}
		\alpha_j(t) + \tx\frac{1}{2\pi} \log \frac{\sqrt{\la}}{2} + \frac{\gamma}{2\pi},	&	\mbox{if } j = k,	\\
		 -\frac{1}{2\pi} K_0(\sqrt{\la} |\yv_j - \yv_k|),	&	\mbox{if } j \neq k.
	\end{cases}
\eeq
In the simplest case of a single point interaction at the origin the boundary condition thus reads
\bdm
	\lim_{\xv \to 0} \phila(\xv) = \lf( \alpha(t) + \tx\frac{1}{2\pi} \log \frac{\sqrt{\la}}{2} + \frac{\gamma}{2\pi} \ri) q(t).
\edm	
The action of $ \ham $ on functions of $ \dom(\ham) $ characterized in the form \eqref{eq: op domain} is
\beq
	\label{eq: op action}
	\lf(\ham +\la \ri) \psi = \lf( - \Delta + \la\ri) \phila
\eeq
and all the information on the interaction is encoded in the boundary conditions. Unlike the case of the quadratic form, the operator domain does depend on time: a generic $ \psi \in \dom(\ham) $ depends on $ t \in \R $ through the regular part $  \phi $ and the charge $ \qv(t) $. The closedness of the form clearly implies the self-adjointness of $ \ham $, provided $ \ham $ is indeed the operator associated with $ \form $, as we are going to show next. A very crucial property of functions in $ \dom(\ham) $ is that in a neighborhood of any point $ \yv_j $ the following asymptotic behavior holds true
\beq
	\label{eq: boundary condition}
	\psi(|\xv|) \underset{|\xv - \yv_j| \to 0}{=} \tx\frac{1}{2\pi} q_j \log\tx\frac{1}{|\xv- \yv_j|} + \alpha_j(t) q_j + o(1),
\eeq
which is indeed the typical way point interactions are defined in the physics literature (see, e.g., \cite{CFT} and references therein).

Notice that, unlike the three-dimensional case, the expression of the form or operator domain for $ \lambda = 0 $ can in principle be obtained by taking the limit $ \lambda \to 0  $ of \eqref{eq: form domain} or \eqref{eq: op domain}, but, due to the singular large-$ |\xv| $ behavior of the Green function at $ \lambda = 0 $, i.e.,  $ \log|\xv| $, such a procedure does not define a well-posed domain decomposition.

For convenience of the reader we recall here how one can heuristically derive the expression of the quadratic form $ \form $ from the formal expression \eqref{eq: formal ham} via a sort of renormalization: pick any function $ \psi $ satisfying the required singular behavior \eqref{eq: boundary condition} at any point $ \yv_j $, then it can be decomposed as $ \psi = \phi_j + \frac{1}{2\pi} q_j \log \frac{1}{|\xv-\yv_i|} $, where $ \phi_j $ remains bounded as $ \xv \to \yv_j $ and 
\beq
	\label{eq: decomposition}
	\phi_{\la}(\xv) \underset{|\xv - \yv_j| \to 0}{=}  \phi_j(\yv_j) + \tx\frac{1}{2\pi} \log \frac{\sqrt{\lambda}}{2} + \tx\frac{1}{2\pi}\gamma - \disp\frac{1}{2\pi}\sum_{k \neq j}  q_k K_0(\sqrt{\la}|\yv_j - \yv_k|) + o(1),	
\eeq
\beq
		\phi_j(\yv_j) = \alpha_j(t) q_j.
\eeq
Now introducing an ultraviolet cut-off $ \eps $ and observing that 
\bdm
	(-\Delta + \la) \psi = (- \Delta  + \la) \phi_{\la},
\edm
if $ |\xv - \yv_j| \geq \eps $ for any $ j = 1, \ldots, N $, one has
\bmln{
	\form[\psi] + \la \lf\| \psi \ri\|_2^2 = \lim_{\eps \to 0} \int_{\cup_k \{ |\xv - \yv_k| \geq \eps \}} \diff \xv \: \psi^*(\xv) \lf[\lf(\tilde H + \la \ri)\psi \ri](\xv) \\
	= \lim_{\eps \to 0} \int_{\cup_k \{ |\xv - \yv_k| \geq \eps \}} \diff \xv \:  \bigg[ \phi_{\la}^*(\xv) + \frac{1}{2\pi} \sum_{j=1}^N q_j^* K_0(\sqrt{\la}|\xv - \yv_j|) \bigg] \lf[ \lf(-\Delta + \la \ri) \phi_{\la}\ri](\xv) 	\\
	= \lf\| \nabla \phi_{\la} \ri\|^2_2 + \la \lf\| \phi_{\la} \ri\|_2^2 + \frac{1}{2\pi} \lim_{\eps \to 0} \sum_{j = 1}^N q_j^* \int_{\cup_k \{ |\xv - \yv_k| \geq \eps \}} \diff \rv \:  K_0(\sqrt{\la}|\xv - \yv_j|) \lf[ \lf(-\Delta + \la \ri) \phi_{\la}\ri](\xv). 
}
The last term can be integrated by parts twice as
\bmln{
	\frac{1}{2\pi} \int_{\cup_k \{ |\xv - \yv_k| \geq \eps \}} \diff \rv \:  K_0(\sqrt{\la}|\xv - \yv_j|) \lf[ \lf(-\Delta + \la \ri) \phi_{\la}\ri](\xv) \\
	= - \frac{1}{2\pi} \sum_{k = 1}^N \int_{\partial \ba_{\eps}(\yv_k)} \diff \sigma \: \phi_{\la}(\xv) \: \nv \cdot \nabla K_0(\sqrt{\la}|\xv - \yv_j|) \\
	= - \frac{1}{2\pi} \int_{\partial \ba_{\eps}(\yv_j)} \diff \sigma \: \phi_{\la}(\xv) \: \nv \cdot \nabla K_0(\sqrt{\la}|\xv - \yv_j|) + o(1)\\
	= \alpha_j(t) q_j + \tx\frac{1}{2\pi} \log \frac{\sqrt{\lambda}}{2} + \tx\frac{1}{2\pi}\gamma - \disp\frac{1}{2\pi} \sum_{k \neq j} q_k  K_0(\sqrt{\la}|\yv_j - \yv_k|)  + o(1),	
}
since  the asymptotics \eqref{eq: k0 asympt} implies
\bdm
	\int_{\partial \ba_{\eps}(\yv_j)} \diff \sigma \: \nv \cdot \nabla K_0(\sqrt{\la}|\xv - \yv_j|) = - \int_{\partial \ba_{\eps}(0)} \diff \sigma \: \frac{1}{\eps} = - 2\pi,
\edm
and the expression of the quadratic form is recovered.


It is worth mentioning that in the two-dimensional case point interactions are always attractive, meaning that there always exists  at least one bound state. For a single point interaction at $ \xv = 0 $ with strength $ \alpha $, its wave function is proportional to $ K_0(\sqrt{\la_{\alpha}}|\xv|) $ and its energy $ E_{\alpha} $ is
\bdm
	E_{\alpha} = - \la_{\alpha} : = - 4 e^{2 \gamma - 4\pi \alpha}.
\edm

\subsection{Time-evolution}

Our goal is to examine the properties of the time-dependent hamiltonians we have just defined and check under which conditions they generate a non autonomous quantum dynamics,  meaning that there exists a two parameter group of unitary operators $U(t,s)$ satisfying, in a sense which has to be specified,  the Schr\"odinger equation
\beq
	\label{eq: generator}
	i \partial_t U(t,s) = \ham U(t,s),
\eeq
in such a way that the function
\beq
	\label{eq: time evolution}
	\psi_t(\xv) = U(t,s) \psi_s(\xv).
\eeq
	solves the Cauchy problem: for any $ \psi \in \dom(H_{\alpha(s)}) $,
\beq
	\label{eq: cauchy}
	\begin{cases}
		i \partial_t \psi_t = \ham \psi_t,	\\
		\psi_s = \psi.
	\end{cases}
\eeq

In this paper the focus of our attention will be on the solution of the time-evolution problem described above.  It is worth mentioning that an explicit expression of the integral kernel of the propagator $ e^{-iH_{\alpha} t} $, when $ \alpha $ does not depend on time, is already known \cite{ABD}, but its extension to the time-dependent case is not straightforward.

Our approach is based on a result, earlier proved and exploited in dimension three, stating that the solution of the Schr\"odinger equation is completely specified by the values of the charges $ q_j (t), \, j= 1,\ldots,N\,$, characterizing the behavior of the wave packet around the scattering centers at each time $t$. The time dependent complex charges are solutions of a system of $N$ coupled Volterra integral equations  -- the {\it charge equations} -- thus reducing the complexity of the problem from the analysis of a non autonomous flow in an infinite-dimensional Hilbert space to the search of solutions to a system of Volterra-type equation for $N$ complex valued functions of time (see \eqref{eq: charge eq} below).  For computational purposes as well as for possible extensions to nonlinear models such a complexity reduction is of course crucial.  The  procedure outlined above has been exploited in \cite{SY}, for the three-dimensional analogue of the problem we are facing here, in \cite{CCLR,CLR,CD,CDFM} to investigate model-atoms ionization triggered by time dependent forces in dimension one and three and in \cite{DFT2} for the derivation of the time-dependent propagator in the case of three-dimensional moving point interactions. For a detailed introduction to the problem considered in this paper as well as many preliminary results we also refer to \cite{A}.

Before stating our main result we need to introduce first some notation: $ U_0(t) $ will denote the free propagator, i.e.,
\beq
	\label{eq: free propagator}
	U_0(t) : = e^{ i \Delta t},
\eeq
with integral kernel for $ t \in \R $ and $ \xv \in \R^2 $
\beq
	U_0(t; |\xv|) = \frac{e^{-\frac{|\xv|^{2}}{4 i t}}}{2 i t}.
\eeq 
The Volterra function of order $ - 1 $ \cite{E} is defined as
\beq
	\label{eq: i}
	\I(t) : = \int_{0}^{\infty}\diff \tau \: \frac{t^{\tau - 1}}{\Gamma(\tau)},
\eeq
where $ \Gamma $ denotes the  Gamma function. Some of the crucial properties of $ \I(t) $ are listed in Section \ref{sec: volterra}. Here we just point out that $ \I $ is an analytic function of $ t $ with branch points at $ 0 $ and $ \infty $.

Before stating our main result we introduce the charge equation associated to the time-evolution of the hamiltonian $ \ham $: given any initial datum $ \psi_s \in \dom(H_{\alv(s)}) $,
\beq
	\label{eq: charge eq}
	\framebox{$\qv(t) + \disp\int_s^t \diff \tau \: \KK(t,\tau) \: \qv(\tau) = \fv(t),$}
\eeq
where 
\beq
	\label{eq: KK}
	\KK_{jk}(t,\tau) : = 
	\begin{cases}
		4 \pi \I(t - \tau) \lf( \alpha_j(\tau) - \tx\frac{1}{2\pi} \log 2 + \frac{\gamma}{2\pi} \ri)	&	\mbox{if } j = k,	\\
		- 2i \I(t-\tau) \disp\int_s^\tau \diff \sigma \:  U_0(\tau-\sigma; |\yv_j - \yv_k|)	&	\mbox{if } j \neq k.
	\end{cases}
\eeq
and
\beq
	\label{eq: fv}
	f_j(t) : = 4\pi \disp\int_s^t \diff \tau \: \I(t-\tau) \lf(U_0(\tau) \psi_s\ri)(\yv_j).
\eeq

	\begin{teo}[Time-evolution]
		\label{teo: evolution}
			\mbox{}	\\
		Let $  U(t,s): L^2(\R^2) \to L^2(\R^2) $ be the map 
		\beq
			\label{eq: ansatz}
			\framebox{$\lf( U(t,s) \psi_s \ri) (\xv) = \lf( U_0(t-s) \psi_s \ri) (\xv) + \disp\frac{i}{2\pi} \sum_{j=1}^N \int_s^t \diff \tau \: U_0\lf(t - \tau; |\xv - \yv_j|\ri) \: q_j(\tau),$}
		\eeq
		where $ \qv(t) $ is a solution of the Volterra-type integral equation \eqref{eq: charge eq}. Then 
		\ben[(a)]
			\item  $ U(t,s) $ is a two-parameter unitary group: for any $ v,t,s \in \R $, $ U(t,s) $ is unitary, $  U(t,t) = \one $ and $ U(t,s) U(s,v) = U(t,v) $;
			\item  for any $ t,s \in \R $, $  U(t,s) $ solves the time-dependent Schr\"{o}dinger equation \eqref{eq: generator}, i.e., for any $ \psi_s \in \dom(H_{\alpha(s)}) $, $ \psi_t : = U(t,s) \psi_s \in \dom(\ham) $ and
				\beq
					\label{eq: td se}
					i \partial_t \psi_t = \ham \psi_t.
				\eeq
		\een
	\end{teo}
	
	\begin{rem}[Uniqueness of $ \qv(t) $]
		\mbox{}	\\
		Although we did not state it explicitly the well-posedness of the time-evolution identified by $ U(t,s) $ requires that the solution to \eqref{eq: charge eq} is unique. This is indeed the case as it is proven in Proposition \ref{pro: uniqueness}.	
	\end{rem}
	
	\begin{rem}[Ansatz \eqref{eq: ansatz}]
		\mbox{}	\\
		The statement of the Theorem \ref{teo: evolution} says that the ansatz \eqref{eq: ansatz} provide the time-evolution of $ \psi_s $ whenever $ \qv(t)$ solves the charge equation \eqref{eq: charge eq}. Such a statement however can be also read in the opposite direction: given $ \psi_t $ any solution to the time-dependent Schr\"{o}dinger equation \eqref{eq: td se}, then it can be rewritten in the form \eqref{eq: ansatz}, with $ \qv(t) $ solving the charge equation \eqref{eq: charge eq}.
	\end{rem}
	
	

\section{Proofs}
\label{sec: proofs}

This Section contains the proofs of the results stated in Theorem \ref{teo: evolution}, which are divided in several steps:
\ben
	\item	first we examine the integral operator defined by the Volterra kernel $ \I $ and prove some of its relevant properties;
	\item  then we focus on the charge equation and prove that there exists a unique solution in the space of continuous functions;
	\item  such information becomes then a crucial ingredient to prove that the form domain $ \dom[\F] $ is invariant under the map $ U(t,s) $;
	\item  next we show that given $ \psi_s \in \dom(H_{\alv(s)}) $ then $ U(t,s) \phi_s \in \dom(\ham) $ and on a dense subset of the Hilbert space $ U(t,s) $ defines an isometry, which coincides with the time-evolution generated by $ \ham $;
	\item finally we show that $ U $ extends to a two-parameter unitary group by density.
\een

\subsection{Properties of the Volterra kernel $ \I $}
\label{sec: volterra}

We start the discussion by recalling some useful properties of the function $ \I[t] $ defined in \eqref{eq: i}. We refer to \cite[Sec. 18.3]{E} (where $ \I(t) $ is denoted as $ \nu(t,-1) $) for further details. One striking relation involving $ \I(t) $ is the inversion formula of the Laplace transform \cite{et},\cite{SKM}: denoting by 
\bdm
	(\LL f)(p) = \int_0^{\infty} \diff t \: e^{-pt} f(t),
\edm 
the usual action of the Laplace transform, then
\begin{equation}\label{eq: invlap}
	\mathcal{L}^{-1}\left(\frac{p}{\log(p)}\right)(t) = \nu \left(t, -1\right) = \I(t).
\end{equation}

Since they will play some role in the following  we also provide the asymptotic expansions of $ \I(t) $ as $ t \to 0 $ or $ t \to \infty $ (see again \cite{E}):
\bdm
	\I(t) \underset{t \to 0}{=} \frac{1}{t \log^2 \left(\frac{1}{t}\right)}\left[1 + \mathcal{O}(\left|\log t \right|^{-1}) \right],
\edm
\bdm
	\I(t) \underset{t \to \infty}{=} e^{t}+\mathcal{O}(t^{-1} ).
\edm
Hence, given the previous expansions,  $ \I(t) \in L^{1}_{\textrm{loc}}(\R) $. 

Next we study the integral operator
\beq
	\label{eq: integral i}
	\lf( I f \ri)(t) : = \int_0^t \diff \tau \: \I(t - \tau) f(\tau).
\eeq
We also denote by $ J $ the integral operator with kernel 
\beq
	\label{eq: integral j}
	\lf( J f \ri)(t) : = \int_0^t \diff \tau \: \J(t - \tau) f(\tau),		\qquad		\mathcal{J}(t-\tau) = - \gamma - \log (t - \tau).
\eeq 
In \cite{CF} the operator $ I $ is investigated in details and several useful properties, such as its smoothing action, are established. Here we only need a notable identity, which is stated in next Lemma \ref{lem: i identity} and a simple estimate of the Sobolev norm of $ I f $ (we refer to \cite{CF} for the proof):

\begin{lem}
		\label{lem: i pro}
		\mbox{}	\\
		If $ f \in H^\nu(0,T) $ with $ 0 < \nu \leq  1 $, then $ I f \in H^\nu(0,T) $, i.e., $ \exists C_t < + \infty $ such that
		\beq
			\label{eq: norm i estimate}
			\lf\| I f \ri\|_{H^{\nu}(0,T)} \leq C_T \lf\| f \ri\|_{H^{\nu}(0,T)}.
		\eeq
		Moreover $ C_T \to 0 $ as $ T \to 0 $.
	\end{lem}
	
	
	\begin{lem}
		\label{lem: i identity}
		\mbox{}	\\
		For any $ t \in \R^+ $ and $ f \in L^1(0,t) $,
		\beq
			\label{eq: i inverse}
			\lf( I J f \ri)(t) = \int_0^t \diff \tau \: f(\tau).
		\eeq
	\end{lem}
	
	\begin{proof}
		We first observe that one has the identity
		\begin{equation}\label{eq: i identity}
			\int_{0}^{t} \diff \tau \: \I(t - \tau) (-\gamma-\log \tau) = 1.
		\end{equation}
		In \cite[Lemma 32.1]{SKM} it is indeed proven that (in the formula proved in the cited Lemma  one should take $ \alpha = 1, h = 0 $)
		\begin{equation}
			\int_{0}^{t} \diff \tau \: \left(\log \tau - \psi(1) \right) \partial_t \nu(t - \tau) = -1,
		\end{equation}
		but, using \cite[Eq. (12), Sect. 18.3]{E}, one can recognize that $ \partial_t \nu(t) = \I(t) $.
		
		Next we note that in the expression
		\bdm
			\lf( I J f \ri)(t) = \int_0^{t} \diff \tau  \int_0^{t-\tau} \diff \sigma \: \I(\tau) \J(t-\tau-\sigma) f(\sigma),
		\edm
		one can exchange the order of the integration, since
		\bmln{
			  \int_0^{t} \diff \tau  \int_0^{t-\tau} \diff \sigma \: \I(\tau) \J(t-\sigma-\tau) f(\sigma)  + \int_0^{t} \diff \sigma  \int_0^{t-\sigma} \diff \tau \: \I(\tau) \J(t-\sigma-\tau) f(\sigma) \\
			  = \int_0^{t} \diff \sigma  \int_0^{t} \diff \tau \: \I(\tau) \J(|t-\sigma-\tau|) f(\sigma) = 2  \int_0^{t} \diff \tau  \int_0^{t-\tau} \diff \sigma \: \I(\tau) \J(|t-\sigma-\tau|) f(\sigma),
		}
		Using \eqref{eq: i identity} we conclude that
		\bdm
			\lf( I J f \ri)(t) = \int_0^{t} \diff \sigma  \int_0^{t-\sigma} \diff \tau \: \I(\tau) \J(t-\sigma-\tau) f(\sigma) = \int_0^{t} \diff \sigma \: f(\sigma).
		\edm
	
 	\end{proof}



\subsection{Derivation of the charge equation}
\label{sec: derivation}

Before starting to discuss the charge equation, we present a heuristic computation which motivates the ansatz \eqref{eq: ansatz}. First of all we set $ s = 0 $ and assume that $  q_j(0) = 0 $. Neglecting any regularity issue, we can compute the time derivative of \eqref{eq: ansatz} and obtain
\bmln{
	i \partial_t \lf( U(t,0) \psi_0 \ri)(\xv) = \lf( - \Delta U_0(t) \psi_0 \ri) (\xv)\\
	  - \frac{1}{2\pi}\sum_{j=1}^N q_j(t)	+ \frac{1}{2\pi} \sum_{j=1}^N \int_0^t \diff \tau \:  \partial_\tau U_0 \lf(t - \tau; |\xv - \yv_j|\ri) \: q_j(\tau) = \\
	 \lf( - \Delta U_0(t) \psi_0 \ri) (\xv)  - \frac{1}{2\pi} \sum_{j=1}^N \int_0^t \diff \tau \:  U_0 \lf(t - \tau; |\xv - \yv_j|\ri) \: \dot q_j(\tau),
}
so that if we take the Fourier transform defined for a function $ f \in L^2(\R^2) $ as
\beq
	\label{eq: fourier}
	\hat{f}(\pv) : = \frac{1}{2\pi} \int_{\R^2} \diff \xv \: e^{-i\pv \cdot \xv} f(\xv),
\eeq
the above expression becomes (we set $ k = |\pv| $)
\beq
	\label{eq: t derivative fourier}
	i \partial_t \lf( \widehat{U(t,0) \psi_0} \ri)(\pv) =  p^2 e^{-ip^2 t} \widehat{\psi_0} (\pv)  - \frac{1}{2\pi} \sum_{j=1}^N \int_0^t \diff \tau \:  e^{- i \pv \cdot \yv_j} e^{-ip^2(t- \tau)} \: \dot q_j(\tau).
\eeq
Similarly, recalling that (see, e.g., \cite[Eq. 6.532.4]{GR})
\beq
	\frac{1}{2\pi} \int_{\R^2} \diff \pv \: \frac{e^{i\pv \cdot \xv}}{p^2 + \la} = \int_0^{\infty} \diff p \: \frac{p J_0(p|\xv|)}{p^2 + \la} = K_0(\sqrt{\la} |\xv|),
\eeq
\bml{
	\label{eq: ham psit fourier}
	\lf(\widehat{\ham  U(t,0) \psi_0}\ri)(\pv) = p^2 \lf( \widehat{U(t,0) \psi_0}(\pv) - \frac{1}{2\pi} \sum_{j=1}^N \frac{q_j e^{-i \pv \cdot \yv_j}}{p^2 + \la} \ri) - \frac{\la}{2\pi} \sum_{j=1}^N \frac{q_j e^{-i \pv \cdot \yv_j}}{p^2 + \la}	\\
	= p^2  e^{-ip^2 t} \widehat{\psi_0} (\pv)  + \frac{1}{2\pi} \sum_{j=1}^N \int_0^t \diff \tau \:  e^{- i \pv \cdot \yv_j} \partial_\tau \lf( e^{-ip^2(t- \tau)} \ri) \: q_j(\tau) - \frac{1}{2\pi} \sum_{j=1}^N  q_j e^{-i \pv \cdot \yv_j}  	\\
	=   p^2 e^{-ip^2 t} \widehat{\psi_0} (\pv)  - \frac{1}{2\pi} \sum_{j=1}^N \int_0^t \diff \tau \:  e^{- i \pv \cdot \yv_j} e^{-ip^2(t- \tau)} \: \dot q_j(\tau),
	}
	which equals \eqref{eq: t derivative fourier}. Therefore,  for any  $ \qv(t) $ and $\psi_0$ for which the right hand side of \eqref{eq: ham psit fourier} is defined, the assumed solution does solve the time-dependent Schr\"{o}dinger equation, at least in a weak sense. To be solution of the charge equation is  the condition that guarantees that for any $ t \in \R $, $ \psi_t \in \dom(\ham) $. Indeed, if we impose the boundary condition as in \eqref{eq: op domain}, we get
\beq
	\label{eq: bc heuristics}
	\frac{1}{2\pi} \int_{\R^2} \diff \pv \: e^{i \pv \cdot \yv_j} \widehat{\phila}(\pv) = \lf( \Gammal \qv \ri)_j(t),
\eeq
and therefore
\bmln{
	\frac{1}{2\pi} \int_{\R^2} \diff \pv \: e^{i \pv \cdot \yv_j} \lf\{ e^{-ip^2 t} \widehat{\psi_0} (\pv) + \frac{i}{2\pi} \sum_{k=1}^N \int_0^t \diff \tau \:  e^{- i \pv \cdot \yv_k} e^{-ip^2(t- \tau)} \: q_k(\tau)	\ri.\\
	\lf. - \frac{1}{2\pi} \sum_{k =1}^N \frac{q_k(t) e^{-i \pv \cdot \yv_k}}{p^2 + \la} \ri\}	
	= \lf( \alpha_j(t) + \tx\frac{1}{2\pi} \log \frac{\sqrt{\la}}{2} - \frac{\gamma}{2\pi} \ri) q_j(t) - \frac{1}{2\pi}\sum_{k \neq j} q_k(t)  K_0(\sqrt{\la} |\yv_j - \yv_k|).
}
The last off-diagonal term cancels exactly and thus the identity becomes
\bmln{
	\frac{1}{2\pi} \int_{\R^2} \diff \pv \: e^{i \pv \cdot \yv_j} \lf\{ e^{-ip^2 t} \widehat{\psi_0} (\pv) + \frac{i}{2\pi} \sum_{k=1}^N \int_0^t \diff \tau \:  e^{- i \pv \cdot \yv_k} e^{-ip^2(t- \tau)} \: q_k(\tau)	\ri.\\
	\lf. - \frac{1}{2\pi} \frac{q_j(t) e^{-i \pv \cdot \yv_j}}{p^2 + \la} \ri\}	
	= \lf( \alpha_j(t) + \tx\frac{1}{2\pi} \log \frac{\sqrt{\la}}{2} + \frac{\gamma}{2\pi} \ri) q_j(t).
}	
Combining the last diverging term on the l.h.s. with the second one, via an integration by parts (here we implicitly assume that the charge belongs to a suitable Sobolev space), we get
\bmln{
	\frac{1}{2\pi} \int_{\R^2} \diff \pv \: \lf\{ e^{i \pv \cdot \yv_j} e^{-ip^2 t} \widehat{\psi_0} (\pv) - \frac{1}{2\pi (p^2+\la)} \int_0^t \diff \tau \:  e^{-ip^2(t- \tau)} \: \lf[ \dot q_j(\tau) - i \la q_j(\tau) \ri]  \ri. \\
	\lf. + \frac{i}{2\pi } \sum_{k \neq j} \int_0^t \diff \tau \:  e^{ i \pv \cdot ( \yv_j - \yv_k)} e^{-ip^2(t- \tau)} \: q_k(\tau) \ri\}	
	= \lf( \alpha_j(t) + \tx\frac{1}{2\pi} \log \frac{\sqrt{\la}}{2} + \frac{\gamma}{2\pi} \ri) q_j(t),
}	
The $\pv$ integral of the second term on the l.h.s. contains an infrared singularity for $ t = \tau $ which goes as $ \log(t - \tau) $: since \cite[Eqs. 3.722.1 \& 3.722.3]{GR}
\bml{
	\int_{\mathbb{R}^{2}} \diff \pv \: \frac{e^{-i p^{2}(t-\tau)}}{p^{2}+\lambda}= - \pi e^{i \lambda(t-\tau)} \lf[ \mathrm{Ci}(\lambda(t-\tau)) -i \mathrm{Si}(\lambda(t-\tau))  \ri] \\
	=    - \pi \lf( \gamma + \log \lambda + \log (t - \tau) \ri) + Q(\lambda; t-\tau)e^{i \lambda (t - \tau)},
}
where $\textrm{Si}( \: \cdot \:)$ and $ \mathrm{Ci}( \: \cdot \:) $ stand for the sine and cosine integral functions \cite[Eqs. 5.2.1 \& 5.2.2]{AS} and (see, e.g., \cite[Eq. 5.2.16]{AS})
\bml{
	Q(\lambda; t - \tau) : = - \pi \lf(1 -e^{i \lambda t} \ri) \lf( \gamma + \log \lambda + \log (t - \tau) \ri) 	\\
	- \pi e^{i \lambda (t - \tau)} \left( \disp\sum_{n=1}^{\infty}\frac{(-(t-\tau)^{2} \lambda^{2})^{n}}{2n(2n)!} - i \textrm{Si}((t - \tau) \lambda) \right).
}
Note that $ Q(0; t - \tau) = 0 $. Hence the charge equation can be rewritten
\bmln{
	\lf(U_0(t) \psi_0\ri)(\yv_j) + \frac{i}{2\pi } \sum_{k \neq j} \int_0^t \diff \tau \:  U_0(t-\tau; |\yv_j - \yv_k|) \: q_k(\tau) 
	- \lf( \alpha_j(t) + \tx\frac{1}{2\pi} \log \frac{\sqrt{\la}}{2} + \frac{\gamma}{2\pi} \ri) q_j(t) \\
	= - \frac{1}{4 \pi} \int_0^t \diff \tau \: \lf( \gamma +\log (t - \tau) + \log\lambda  - \tx\frac{1}{\pi} Q(\lambda; t-\tau) \ri)  \partial_{\tau} \lf( e^{i \lambda (t - \tau)} q_j(\tau) \ri) 
	} 
and taking the limit $ \lambda \to 0 $ (notice the exact cancellation of the diverging $ \log \lambda $ terms)
\bml{
	\label{eq: charge eq heuristics}
	\lf(U_0(t) \psi_0\ri)(\yv_j) + \frac{i}{2\pi } \sum_{k \neq j} \int_0^t \diff \tau \:  U_0(t-\tau; |\yv_j - \yv_k|) \: q_k(\tau) 
	- \lf( \alpha_j(t) - \tx\frac{1}{2\pi} \log 2 + \frac{\gamma}{2\pi} \ri) q_j(t) \\
	= - \frac{1}{4 \pi} \int_0^t \diff \tau \: \lf( \gamma +\log (t-\tau)    \ri)  \dot q_j(\tau).
	} 
	If we now apply to both sides the integral operator $ I $ defined in \eqref{eq: integral i} and exploit the property proven in Lemma \ref{lem: i identity}, we finally recover the charge equation \eqref{eq: charge eq}.

\subsection{Charge equation} We now consider the charge equation and its solution.

	\begin{pro}[Existence and uniqueness of solutions to \eqref{eq: charge eq}]
		\label{pro: uniqueness}
		\mbox{}	\\
		Given $\psi_{s}\in \dom(H_{\alv(s)}) $, the solution $ \qv(t)$ of the charge equation \eqref{eq: charge eq} exists and is unique in $ C(0,T) $ for any $T < \infty$. Moreover $ \qv(t) \in  H^{\nu}(0,T) $ for any $ \nu < 3/4 $.
	\end{pro}


	\begin{proof}	
		According to the general theory of Volterra integral equations \cite{M}, specialized to the linear case there exists at least one continuous solution to \eqref{eq: charge eq} and it is unique if the following conditions are satisfied:
	
	\ben[(a)]
			\item $ \fv(t) $ is continuous on $ t \in \R^+ $;
			\item $ \KK(t,\tau) $ is measurable and $\KK(t, \cdot) \in L^1(0,t) $ for any finite $ t $;
		\een
Let us   check that the integral kernel $ \KK_{jk}(t, \tau) \in L^1(0,t) $ 
		In particular  to see that $ \KK_{jk}(t,\tau) $ is integrable, it suffices to notice that the diagonal term is $ L^1 $ because $ \I(t-\tau) $ is, while for $ j \neq k $
		\bdm
			\lf| \KK_{jk}(t,\tau) \ri| \leq  2 \lf| \I(t-\sigma) \ri| \lf|\disp\int_s^\tau \diff \sigma \:  U_0(\tau-\sigma; |\yv_j - \yv_k|)\ri| \leq C_t \lf| \I(t-\sigma) \ri|,
		\edm
		for some $ C_t < \infty $, if $ t < \infty $. Indeed the integral $$ \int_s^\tau \diff \sigma \:  U_0(\tau-\sigma; |\yv_j - \yv_k|)$$
	is explicitly computable and it is finite for any $ \tau\geq 0 $, if $ |\yv_j - \yv_k| > 0 $.

		Therefore to complete the proof we need to show that $ \fv(t) $ is a continuous function of $ t $ in any compact subset of $ \R^+ $. By hypothesis $ \psi_s \in \dom(H_{\alv(s)}) $ and thus it can be decomposed as in \eqref{eq: op domain}, i.e.,
		\bml{
			\label{eq: fjt}
			f_j(t) = 4 \pi \int_{s}^{t} \diff \tau \: \I(t - \tau) \bigg\{ \lf(U_{0}(\tau) \philas \ri)(\yv_j) + 2 q_j(s) \int_{\R^2} \diff \xvp \: U_{0}(\tau;|\yv_j - \xvp|) K_0(\xvp - \yv_j)  \\
			 + 2 \sum_{k \neq j} q_k(s) \int_{\R^2} \diff \xvp \: U_{0}(\tau;|\yv_j - \xvp|) K_0(\xvp - \yv_k) \bigg\}
		}
		with $ \philas \in H^2(\R^2) $.
	
		We first consider the term involving $ \philas $: applying the Fourier transform, we have
		\bmln{
			4 \pi \lf(U_{0}(\tau) \philas\ri)(\yv_j) = 2 \int_{\R^2} \diff \pv \: e^{i\pv \cdot \yv_j} e^{-i p^2\tau} \widehat{\philas}(\pv)
			= 2 \int_{\mathbb{R}^{2}} \diff \pv \: e^{-i p^2\tau} \widehat{\lf(T_{\yv_j}^{-1} \philas\ri)}(\pv) \\
			= 2\pi \int_{0}^{\infty} \diff \varrho \: e^{-i \varrho \tau} \mean{\widehat{\lf(T_{\yv_j}^{-1} \philas\ri)}}(\sqrt{\varrho}) 
			= (2\pi)^{3/2} \lf(\F G(\varrho) \ri)(t)
		}
		where 
		\bdm
			G_1(\varrho) : = \one_{[0,+\infty)}(\varrho) \: \mean{\widehat{\lf(T_{\yv_j}^{-1} \philas\ri)}}(\sqrt{\varrho}),
		\edm
		$ \F $ stands for the Fourier transform in $ L^2(\R) $ and we have denoted by $ \mean{f} $ the angular average of a function on $ \R^2 $, i.e.,
		\beq
			\mean{f}(p) = \frac{1}{2\pi} \int_0^{2\pi} \diff \vartheta \: f(p,\vartheta).
		\eeq
		In order to bound the norm of $ \lf(U_{0}(\tau) \philas\ri)(\yv_j) $ in $ H^{\nu}(\R) $, we estimate
		\bml{
			4\pi \int_{\mathbb{R}} \diff \varrho \: \lf|\varrho\ri|^{2\nu} \lf| \left( \mathcal{F} \lf(U_{0}(\: \cdot \:) \philas\ri)(\yv_j)  \right)(\varrho)\right|^{2} \\
			= 32 \pi^4 \int_{\mathbb{R}} \diff \varrho \: \lf|\varrho\ri|^{2\nu} \lf| \left( \mathcal{F}^{-1} (\F G_1)  \right)(-\varrho)\right|^{2} \\
			= 64 \pi^4 \int_{0}^{\infty} \diff p \: p^{4\nu+1} \lf| \mean{\widehat{\lf(T_{\yv_j}^{-1} \philas\ri)}}(p) \right|^{2} \leq C \int_{\R^2} \diff \pv \: p^{4\nu} \lf| \widehat{\lf(T_{\yv_j}^{-1} \philas\ri)}(\pv) \right|^{2}.
		}
		Since $ \philas \in H^2 $, the last integral is bounded for any $ 0 \leq \nu \leq 1 $ and therefore $ \lf(U_{0}( \: \cdot \:) \philas\ri)(\yv_j) \in H^1(0,T) $. Thanks to Lemma \ref{lem: i pro} the Sobolev degree is conserved by the action of $ I $ and therefore the first term in \eqref{eq: fjt} is in $ H^1(0,T) $, which implies, via Sobolev inequality, that it is a continuous function of $ t \in \R^+ $.

		Let us consider now the sum in the last term in \eqref{eq: fjt}: we first rewrite
		\bml{
				\lf( U_0(\tau) T_{\yv_k} K_0(\sqrt{\la} \: \cdot \:) \ri)(\yv_j) = \int_{\R^2} \diff \xvp \: U_{0}(\tau;|\yv_j - \xvp|) K_0(\xvp - \yv_k)  \\
		=  \frac{1}{2 \pi} \int_{\R^2} \diff \pv \: \frac{e^{i \pv \cdot (\yv_j - \yv_k)} e^{-i p^2 \tau}}{p^2 + \la} =  \frac{1}{2 \pi} \int_{0}^{\infty} \diff p \: p \int_0^{2\pi} \diff \vartheta \:\frac{e^{i p |\yv_j - \yv_k| \cos \vartheta} e^{-i p^2 \tau}}{p^2 + \la} \\ = \int_{0}^{\infty} \diff p \: p \: \frac{  e^{-i p^2 \tau}J_0(p |\yv_j - \yv_k|)}{p^2 + \la} = \pi \lf(\F G_2\ri)(\tau),
			}
		with 
		\bdm
			G_2(\varrho) = \one_{[0,+\infty)}(\varrho) \frac{J_0(\sqrt{\varrho} |\yv_j - \yv_k|)}{\varrho + \la}.
		\edm
		As before in order to bound the $ H^{\nu}$-norm of $ ( U_0(\tau) T_{\yv_k} K_0(\sqrt{\la} \: \cdot \:) )(\yv_j) $ we estimate, using the asymptotics of Bessel functions for large argument (see \cite[Eq. 9.2.1]{AS})
		\bml{
			\int_{\mathbb{R}} \diff \varrho \: \left|\varrho \ri|^{2\nu} \lf| \lf(\mathcal{F} \lf( U_0(\tau) T_{\yv_k} K_0(\sqrt{\la} \: \cdot \:) \ri)(\yv_j) \ri)(\varrho) \ri|^2 = \pi^2 \int_{\mathbb{R}} \diff \varrho \: \left|\varrho \ri|^{2\nu} \lf| \lf(\mathcal{F}^{-1} \F G_2 \ri)(-\varrho)  \ri|^2	\\
			= \pi^2 \int_{0}^{\infty} \diff p \: p^{4\nu +1}  \frac{J^2_0(p |\yv_j - \yv_k|)}{\lf(p^2 + \la\ri)^2} \leq C \int_{0}^{\infty} \diff p \: p^{4\nu+1} \frac{1}{(p + 1)\lf(p^2 + \la\ri)^2},
		}
		which is finite for any $ \nu < 3/4 $. This implies that $ ( U_0(\tau) T_{\yv_k} K_0(\sqrt{\la} \: \cdot \:) )(\yv_j) $ belongs to $ H^{\nu}(0,T) $, for any $ \nu < 3/4 $. The action of $ I $ does not change the Sobolev degree and therefore the last term in \eqref{eq: fjt} is continuous in $ t $, thanks again to Sobolev inequality.

		The last term to consider is the second one in \eqref{eq: fjt}: as before we have
		\bml{
			\lf( U_0(\tau) T_{\yv_j} K_0(\sqrt{\la} \: \cdot \:) \ri)(\yv_j) = \int_{\R^2} \diff \xvp \: U_{0}(\tau;|\yv_j - \xvp|) K_0(\xvp - \yv_j)  
			=  \frac{1}{2} \int_{0}^{\infty} \diff \varrho \: \frac{e^{-i \varrho \tau}}{\varrho + \la}	\\ 
			= \tx\frac{1}{2} e^{i \lambda \tau} \lf(i \mathrm{Si}(\lambda \tau) - \mathrm{Ci}(\lambda \tau) \ri) =   - \tx\frac{1}{2} \lf( \gamma + \log \lambda + \log \tau \ri) + \tx\frac{1}{2\pi} Q(\lambda; \tau)e^{i \lambda \tau},
	}
		Hence
		\beq
			I \lf[ \lf( U_0(\tau) T_{\yv_j} K_0(\sqrt{\la} \: \cdot \:) \ri)(\yv_j) \ri] = 1 + \int_s^t \diff \tau \: \I(t - \tau) \lf[ -\tx\frac{1}{2} \log \lambda + \tx\frac{1}{2\pi} Q(\lambda; \tau) \ri]
		\eeq
		but since $ Q(\lambda,\tau) $ is a smooth function on any compact set, the same applies to the second term in \eqref{eq: fjt}, which is continuous as well.
		
		In order to prove the last statement it suffices to apply a bootstrap like argument: for sufficiently small times the charge equation can be solved since the operator $ 1 + \mathcal{K} $ is invertible in $ H^{\nu}(0,T) $, $ \nu < 3/4 $. This is a consequence of \eqref{eq: norm i estimate} and differentiability of $ \alpha(t) $. Hence $ q \in H^\nu(0,T) $, $ \nu < 3/4 $, for small enough $ T $, but then one can repeat the argument with initial condition $ q(T) $ so proving the statement.		
	\end{proof}

\subsection{Time-evolution in the form and operator domains}

In this Section we show that the form domain is invariant under $ U(t,s) $. 

	\begin{pro}[Invariance of \mbox{$ \dom[ \F ] $} for initial data in $ \dom(H_{\alv(s)}) $]
		\label{pro: invariance}
		\mbox{}	\\
		Let $ \qv(t) $ be the unique solution to \eqref{eq: charge eq} with initial condition $ \lf. \qv(t) \ri|_{t = s} = \qv(s) $ and $ \psi_s \in \dom(H_{\alv(s)}) $, then $ U(t,s) \psi_s \in \dom[\F] $ for any $ t \in \R $.
	\end{pro}

	\begin{proof}
		In order to prove the statement we need to show that 
		\bdm
			\lf(U(t,s)\psi_s\ri)(\xv) - \frac{1}{2\pi} \sum_{j =1}^N q_j(t) K_0(\sqrt{\la}|\xv - \yv_j|) \in H^1(\R^2),
		\edm
		whenever $ \psi_s \in \dom[\F] $. Setting for simplicity $ s = 0 $ and passing to the Fourier representation, this is equivalent to require that the following function of $ \pv $
		\beq
			e^{-i p^2 t} \widehat{\psi_0}(\pv) + \frac{i}{2\pi} \sum_{j=1}^N \int_0^t \diff \tau \: e^{i \pv \cdot \yv_j} e^{-i p^2 (t - \tau)} q_j(\tau) - \frac{1}{2\pi} \sum_{j=1}^N \frac{q_j(t) e^{i \pv \cdot \yv_j}}{p^2 + \la}	
		\eeq
		belongs to $ L^2(\R^2, (p^2+1)\diff \pv) $. After an integration by parts the above expression becomes (here $ \dot q_j $ stands for the weak derivative of $ q_j $, which belongs at least to $ H^{\nu}(0,T) $, $ \nu < -1/4 $, since $ q_j(t) \in H^{\nu} $, $ \nu < 3/4 $ by Proposition \ref{pro: uniqueness})
		\bml{
			\label{eq: decomposition psit}
			e^{-i p^2 t} \widehat{\psi_0}(\pv) - \frac{1}{2\pi} \sum_{j=1}^N \frac{q_j(0) e^{i \pv \cdot \yv_j} e^{-ip^2 t}}{p^2 + \la} \\
			+ \frac{1}{2\pi (p^2 + \la)} \sum_{j=1}^N \int_0^t \diff \tau \: e^{i \pv \cdot \yv_j} e^{-i (p^2 + \la) (t - \tau)} \partial_\tau \lf( e^{i\la \tau} q_j(\tau) \ri).
		}
		Now the first two terms represent the free evolution of the regular part $ \phi_{\la,0} $ of the initial state $ \psi_0 $. Since by hypothesis $ \psi_0 \in \dom[\F] $, then $ \phi_{\la,0} \in H^1(\R^2) $ and therefore the sum of those two terms (or rather their Fourier anti-transform) belongs to $ H^1(\R^2) $ as well.

		It remains then to prove that the last term in \eqref{eq: decomposition psit} is in $ L^2(\R^2, (p^2+1)\diff \pv)  $: setting $ z(t) = \partial_t (e^{i\la (t-\tau)} q_j(t)) $ and calling each term of the sum $ g_j(\pv) $ for short, we have
		\bdm
			\int_{\R^2} \diff \pv \: (p^2 +1) \lf| g_j(\pv) \ri|^2 = \frac{1}{8} \int_{0}^{\infty} \diff \varrho \: \frac{\varrho + 1}{(\varrho + \la)^2} \: \lf| \lf( \F z_{0,t} \ri)(\varrho) \ri|^2 
		\edm
		where we have denote  for any function $ f: [0,T] \to \C $ and $ 0 \leq a < b \leq T $,
		\beq
			f_{a,b}(t) : = f(t) \one_{[a,b]}(t).
		\eeq
		Now the r.h.s. is obviously bounded by 
		\bdm
			C \lf\| z_{0,t} \ri\|_{H^{-1/2}(\R)}^2 \leq C_t \lf( \lf\| \lf(q_j\ri)_{0,t} \ri\|_{H^{-1/2}(\R)}^2 + \lf\| \lf(\dot q_j\ri)_{0,t} \ri\|_{H^{-1/2}(\R)}^2 \ri),
		\edm
		for some $ C_t < \infty $ for finite $ t $.
		In \cite[Lemma 2.1]{CFNT} it is proven that if $ f \in H^{\nu}(0,T) $ for $ 0 < \nu < 3/2 $ but $ \nu \neq 1/2 $ (see also \cite[Lemma 5]{AT}), then $ f_{0,t} \in H^\nu(\R) $. Therefore the first term on the r.h.s. of the above expression is always bounded, since by Proposition \ref{pro: uniqueness}, $ q_j \in H^{\nu}(0,T) $ for any $ 1/2 < \nu < 3/4 $. For the second term one can not apply directly  \cite[Lemma 2.1]{CFNT} because the Sobolev degree of $ \dot q_j $ is negative, but one can circumvent such a problem by modifying the extension of $ \dot q_j $ \cite[Proof of Theorem 4]{ADFT}: let $ Q_{j}(\tau) $, $ \tau \in \R $, be the following function which extends $ q_j(\tau) $
		\bdm
			Q_j(\tau) = 
			\begin{cases}
				q_j(0)		&		\mbox{for } \tau < 0,	\\
				q_j(t)		&		\mbox{for } \tau > t,		\\
				q_j(\tau)	&		\mbox{for } 0 \leq \tau \leq t.
			\end{cases}
		\edm
		Then one has that $ \dot Q_j = \lf(\dot q_j\ri)_{0,t} $ but $ Q_j \in H^{\nu}_{\mathrm{loc}}(\R) $  for any $ 1/2 < \nu < 3/4 $, which implies that $ \dot Q_j \in H^{\nu-1}(\R) $, since it is compactly supported. In conclusion $ \lf(\dot q_j\ri)_{0,t} \in H^{\nu}(\R) $ for any $ -1/2 < \nu < -1/4 $ and the second term is bounded as well.
	\end{proof}
	
\noindent
The above result in combination with the heuristic computation made at the beginning of Section \ref{sec: derivation} yields the following very important.

	\begin{cor}
		\label{cor: weak se}
		\mbox{}	\\
		Let $ \psi_s \in \dom[\F] $, then $ U(t,s) \psi_s $ solves the time-dependent Schr\"{o}dinger equation \eqref{eq: generator} in the quadratic form sense, i.e., for any $ \phi \in \dom[\F] $,
		\beq
			\label{eq: weak se}
			i \partial_t \braket{\phi}{U(t,s) \psi_s} = \form[\phi,U(t,s) \psi_s],
		\eeq
		for any $ t \in \R $ and with $ \form[ \, \cdot \, , \, \cdot \,] $ standing for the sesquilinear form associated to the quadratic form $ \form[ \, \cdot \, ]$.
	\end{cor}
	
	\begin{proof}
		It suffices to note that the identities proven at the beginning of Section \ref{sec: derivation}, up to \eqref{eq: ham psit fourier}, are in fact rigorous once projected onto a state $ \phi \in \dom[\F] $. The result of Proposition \ref{pro: invariance} completes the argument.
	\end{proof}
	
	\subsection{Completion of the proof}

	In order to complete the proof of Theorem \ref{teo: evolution}, we have to show that $ U(t,s): \dom(H_{\alv(s)}) \to \dom(\ham) $ and it is an isometry in that subspace.
	
	\begin{lem}
		\label{lem: op domain invariant}
		\mbox{}	\\
		Let $ \psi_s \in \dom(H_{\alv(s)}) $, then for any $ t \in \R $, $ U(t,s) \psi_s \in \dom(\ham) $.
	\end{lem}
	
	\begin{proof}
		Thanks to Proposition \ref{pro: invariance} at least $ U(t,s) \psi_s \in \dom[\F] $ so that it can be decomposed in a regular part in $ H^1(\R^2) $ plus the singular terms given in \eqref{eq: form domain}. However this is the only information needed to make rigorous the heuristic derivation presented from eqs. \eqref{eq: bc heuristics} to \eqref{eq: charge eq heuristics}. The property stated in Lemma \ref{lem: i identity} and the fact that $ \qv(t) $ solves the charge equation implies then the result.
	\end{proof}
	
	\begin{lem}
		\label{lem: isometry}
		\mbox{}	\\
		Let $ \psi_s \in \dom(H_{\alv(s)}) $, then 
		\beq
			\label{eq: isometry}
			\lf\| U(t,s) \psi_s \ri\|_2 = \lf\| \psi_s \ri\|_2.
		\eeq
	\end{lem}
	
	\begin{proof}
		We simply compute the time-derivative of the $ L^2 $ norm of the ansatz \eqref{eq: ansatz}:
		\beq
			\partial_t \lf\| U(t,s) \psi_s \ri\|^2_2 = 2 \Re \braket{U(t,s) \psi_s}{\partial_t U(t,s) \psi_s} = - 2 \Re \lf(  i \form[U(t,s) \psi_s,U(t,s) \psi_s] \ri) = 0,
		\eeq
		thanks to Lemma \eqref{lem: op domain invariant}, Corollary \ref{cor: weak se} and the trivial observation that if $ \psi_t \in \dom(\ham)$ then $ \psi_t \in \dom[\F] $.
	\end{proof}
		 
	\begin{proof}[Proof of Theorem \ref{teo: evolution}]
		Lemma \ref{lem: op domain invariant} in combination with Corollary \ref{cor: weak se} implies that given any $ \psi_s \in \dom(H_{\alv(s)}) $, $ U(t,s) \psi_s $ solves the time-dependent Schr\"{o}dinger equation. 
		
		Moreover $ U(t,t) = \one $ and, for any $ \psi_0 \in \dom(H_{\alpha(0)}) $, 
		\bml{
			\lf( U(t,s) U(s,0) \psi_0 \ri)(\xv) = \lf(U_0(t) \psi_0 \ri)(\xv) + \disp\frac{i}{2\pi} U_0(t-s) \sum_{j=1}^N \int_0^s \diff \tau \: U_0\lf(s - \tau; |\xv - \yv_j|\ri) \: q_j(\tau) \\
			+ \disp\frac{i}{2\pi} \sum_{j=1}^N \int_s^t \diff \tau \: U_0\lf(t - \tau; |\xv - \yv_j|\ri) \: q_j(\tau) \\
			= \lf(U_0(t) \psi_0 \ri)(\xv)
			+ \disp\frac{i}{2\pi} \sum_{j=1}^N \int_0^t \diff \tau \: U_0\lf(t - \tau; |\xv - \yv_j|\ri) \: q_j(\tau) = \lf( U(t,0) \psi_0 \ri)(\xv),
		}
		i.e., the map $ U(t,s) $ satisfies the group composition rules. Since $ \dom(\ham) $ is densely defined, one can extend the map $ U(t,s) $ to the whole Hilbert space by density and, due to the properties above, such an extension is automatically unitary.
	\end{proof}

\medskip
	\noindent
	{\bf Acknowledgements.} R.C. and M.C. acknowledge the support of MIUR through the FIR grant 2013 ``Condensed Matter in Mathematical Physics (Cond-Math)'' (code RBFR13WAET). The authors also thank \textsc{A. Teta} for helpful discussions about the presentation of the model.

\end{document}